\documentclass[11pt]{article}
\usepackage{amsmath,amssymb,amsthm,amsfonts,mathptmx,mathrsfs,fullpage}
\usepackage{tablefootnote}
\usepackage[margin=1in]{geometry}
\usepackage{verbatim}

\usepackage{graphicx,color,cases}
\usepackage{hyperref}  
\usepackage{xfrac}  
\definecolor{DarkGreen}{rgb}{0.1,0.5,0.1}
\definecolor{DarkRed}{rgb}{0.5,0.1,0.1}
\definecolor{DarkBlue}{rgb}{0.1,0.1,0.5}
\hypersetup{
    unicode=false,              pdftoolbar=true,            pdfmenubar=true,            pdffitwindow=false,          pdfnewwindow=true,          colorlinks=true,           linkcolor=DarkBlue,              citecolor=DarkGreen,            filecolor=DarkGreen,          urlcolor=DarkBlue,                          pdftitle={Edge-isoperimetric inequalities for the symmetric product of graphs},
    pdfauthor={Yingkai Ouyang},
        }

\newtheorem{theorem}{Theorem}[section]

\newtheorem{lemma}[theorem]{Lemma}

\theoremstyle{definition}

\newtheorem{example}[theorem]{Example}

\numberwithin{equation}{section}

\DeclareMathOperator{\wt}{wt} 
 
\def\floor#1{{\left \lfloor {#1} \right \rfloor  }}

\def\>{\rangle} 
\def\<{\langle}

\begin{document}

\title{Permutation-invariant qudit codes from polynomials} 
\author{Yingkai Ouyang  
\footnote{yingkai\_ouyang@sutd.edu.sg 
\newline 
{\em 2010 Mathematics Subject Classification.}  Primary 81P70, 94B25; Secondary 05A10, 94B60.
\newline
{\em Key words and phrases}. Quantum coding, combinatorial codes.}
} 
\date{
Singapore University of Technology and Design, \\ 
8 Somapah Road, Singapore 487372.\\
Centre for Quantum Technologies, National University of Singapore\\ 
3 Science Drive 2, Singapore 117543.\\
\today}

\maketitle

\begin{abstract}
A permutation-invariant quantum code on $N$ qudits is any subspace stabilized by the matrix representation of the symmetric group $S_N$ as permutation matrices that permute the underlying $N$ subsystems.
When each subsystem is a complex Euclidean space of dimension $q \ge 2$, 
any permutation-invariant code is a subspace of the symmetric subspace of $(\mathbb C^q)^N.$
We give an algebraic construction of new families of of $d$-dimensional permutation-invariant codes
on at least $(2t+1)^2(d-1)$ qudits that can also correct $t$ errors for $d \ge 2$.
The construction of our codes relies on a real polynomial with multiple roots at the roots of unity, 
and a sequence of $q-1$ real polynomials that satisfy some combinatorial constraints.
When $N > (2t+1)^2(d-1)$, we prove constructively that an uncountable number of such codes exist.
\end{abstract}   

\maketitle

\section{Introduction} \label{sec:intro}  

The theory of quantum information and linear algebra is intimately related.
The fundamental entities in quantum information theory are quantum states and quantum processes, and they both can be defined in the language of linear algebra.
States on physical systems can be represented either as vectors of unit norm or positive semidefinite linear operators of unit trace on a complex Hilbert space,
while quantum processes are linear maps from quantum states to quantum states.
In this paper, we only need to consider $N$-qudit quantum states with Hilbert space $\mathcal H = (\mathbb C^q)^{ \otimes N}$ with $q\ge 2$, and quantum processes that map $N$-qudit states to $N$-qudit states,
where the Hilbert space of a single qudit is $\mathbb C^q$.
Here, $N$ is the number of qudits in the physical system that we consider.
When $q=2$ and $q=3$, we say that there are $N$ qubits and $N$ qutrits in the physical system respectively.

In quantum information theory, a quantum state is either a pure state or a probabilistic ensemble of pure states known as a mixed state, with an associated complex Hilbert space $\mathcal H$.
A pure state is a vector in $\mathcal H$ with unit norm;
a mixed state is a positive semidefinite operator with unit trace in $L(\mathcal H)$.
In the spectral decomposition of a mixed state $\rho = \sum_{i} p_i \psi_i \psi_i^*$, its eigenvectors of unit norm $\psi_i$ and eigenvalues $p_i$ correspond to the underlying pure states and their associated probabilities respectively.
A quantum state represented as a linear operator is known as a density operator, 
and we denote the set of all density operators on $\mathcal H$ as $\mathfrak D(\mathcal H)$. The set $\mathfrak D(\mathcal H)$ is isomorphic to the set of all positive semi-definite operators in $L(\mathcal H)$ with unit trace,
and this allows one to have a purely linear algebraic interpretation of any quantum state.

In this paper, we represent a quantum process with a quantum channel, which maps a density operator to a density operator.
The theory of quantum channels is well-studied, and the characterization of quantum channels as a completely positive and trace-preserving maps dates back to the work of Choi and others.
Namely, a linear map $\Phi :  L(\mathcal H) \to L(\mathcal H)$ is a quantum channel if and only if 
it is completely positive and trace-preserving (CPT).
It is often convenient to utilize the non-unique Kraus representation of a quantum channel \cite{HeK69,HeK70,kraus}. 
Namely, for any quantum channel $\Phi : L(\mathcal H) \to L(\mathcal H)$, 
there exist linear operators $A_i \in L(\mathcal H)$ such that for every $\rho \in L(\mathcal H)$, 
\[\Phi(\rho) = \sum_{i} A_i \rho A_i^*\] 
and  $\sum_{i} A_i^* A_i $ is the identity operator on $\mathcal H$.
The linear operators $A_i$ are known as Kraus operators of $\Phi$.

The theory of quantum error correction is a subfield of quantum information theory where the robustness of certain families of quantum states under certain families of quantum processes is studied.
The fundamental object in quantum error correction is the quantum code $\mathcal C$, 
which is a $d$-dimensional subspace of $\mathcal H$. 
A possible goal in quantum error correction can be to find quantum codes $\mathcal C$ for a fixed quantum channel $\mathcal N$.
When there exists a quantum channel $\mathcal R$ 
such that for every density operator $\rho$ supported on $\mathcal C$, 
\begin{align}
\mathcal R( \mathcal N (\rho)) = \rho \label{eq:perfect-qecc},
\end{align}
we say that perfect quantum error correction is possible using the quantum code $\mathcal C$ with respect to the channel $\mathcal N$.

The necessary and sufficient conditions for perfect quantum error correction was originally proved by Knill and Laflamme \cite{KnL97}. We now rephrase this foundational result in quantum error correction using the language of linear algebra. 
\begin{theorem}[Knill-Laflamme \cite{KnL97}]
\label{thm:KL} 
Let $\mathcal H$ be a complex Euclidean space, and let $\mathcal C$ be a $d$-dimensional subspace of $\mathcal H$ with orthonormal basis vectors $\psi_1, \dots, \psi_d$.
Let $\mathcal N$ be a quantum channel with Kraus operators $A_i$. 
Suppose that for all $i,j$ there exist complex numbers $g_{i,j}$ such that the following conditions hold.
\begin{enumerate}
\item Orthogonality conditions:  $\psi_a^* A_i ^* A_j \psi_b = 0$ for all $a \neq b$.
\item Non-deformation conditions: $\psi_a^* A_i ^* A_j \psi_a = g_{i,j}$ for all $a = 1,\dots, d$.
\end{enumerate}
Then for every density operator $\rho$ supported on $\mathcal C$, there exists a quantum channel $\mathcal R$ such that 
$\mathcal R (\mathcal N(\rho))  =\rho $.
Moreover, $\mathcal R(\rho) = \sum_i R_i \rho R_i ^*$ where
$R_i = P U_i^*$
where $P= \sum_{a} \psi_a \psi_a^*$ is the projector onto $\mathcal C$, 
and $U_i$ is the unitary in the polar decomposition of $A_i P = U_i \sqrt{P A_i^* A_i P}$.
\end{theorem}
Because of the tensor product structure of $\mathcal H = (\mathbb C_q)^{\otimes N}$, 
any linear operator $A \in L(\mathcal H)$ admits the decomposition
$A = \sum_{j}A_{j,1} \otimes \dots \otimes A_{j,N}$. 
Let $I_q \in L(\mathbb C^q)$ denote a size $q$ identity matrix.
Then, using the decomposition of $A$, we can calculate the weight of $A$, which we define as 
\[
\wt(A) =N- \max_{S \subseteq \{1,\dots, N\}}\{ |S| : A_{j,k} = I_q \}.
\]
The weight of $A$ is the number of qudits on which $A$ acts non-trivially.
When the weight of every Kraus operator of a quantum channel is zero, the quantum channel is the identity channel.
We say that a quantum channel introduces $t$ errors if (1) the quantum channel can have Kraus operators with maximum weight t and (2) the quantum channel cannot have Kraus operators with maximum weight strictly less than $t$.  
The quantum code $\mathcal C$ corrects $t$ errors if perfect quantum error correction is possible for the quantum code $\mathcal C$ with respect to all quantum channels that introduce up to $t$ errors.
The distance of a quantum code $\mathcal C$ with basis vectors $\psi_1,\dots, \psi_d$ is the smallest integer $D$, such that for every linear operator $A \in L(\mathcal H)$ with weight strictly less than $D$,
we have (i) $\psi_a^* A \psi_b = 0$ for all $a \neq b$
and (ii) $\psi_1^* A \psi_1 = \dots = \psi_d^* A \psi_d $.
To show that a quantum code corrects $t$ errors, it suffices to show that the quantum code has a distance of at least $2t+1$.

In this paper, we are interested in quantum codes that not only correct $t$ errors, but are also invariant under the permutation of any of the underlying subsystems. 
To make precise this notion of permutation-invariance, 
note that $\mathcal H = (\mathbb C^q)^{\otimes N}$ has a tensor product structure,
and its standard basis vectors are 
\[
{\bf e}_{a_1}\otimes \dots \otimes {\bf e}_{a_N},
\]
where $a_j=1,\dots,q$ for all $j=1,\dots , N$.
and ${\bf e}_1,\dots,{\bf e}_q$ are standard basis column vectors of $\mathbb C^q$.
Consider a representation $\varphi : S_N \to L(\mathcal H)$ of the symmetric group $S_N$ such that for all $P\in S_N$, 
\begin{align}
\varphi(P) ({\bf e}_{a_1}\otimes \dots \otimes {\bf e}_{a_N})
=
{\bf e}_{P(a_1)}\otimes \dots \otimes {\bf e}_{P(a_N)}.
\end{align}
A quantum code $\mathcal C$ is permutation-invariant if for every $P \in S_N$ and every $\psi \in \mathcal C$,
\begin{align}
\varphi ( P ) \psi = \psi.
\end{align}
The stabilizer of a quantum code $\mathcal C$ is the maximal subset $S \subseteq L(\mathcal H)$ such that for all $\sigma \in S$
and $\psi \in \mathcal C$, we have
$\sigma \psi = \psi$.
Additive quantum codes are the quantum codes for which their stabilizer is abelian. 
Since $S_N$ is non-abelian for $N \ge 3$, permutation-invariant quantum codes are necessarily non-additive codes for $N \ge 3.$
The theory of additive quantum codes has benefited much from its one-to-one correspondence to theory of finite fields \cite{calderbank1998quantum,Rai99NQC}, but such a direct connection is unlikely to exist in the case of non-additive quantum codes such as permutation-invariant quantum codes.

In this paper, we give an algebraic construction of new families of of $d$-dimensional permutation-invariant codes
on at least $(2t+1)^2(d-1)$ qudits that can also correct $t$ errors,
for any $q \ge 2$.
Our codes are also non-additive and non-binary quantum codes,
and their construction relies on a real polynomial $f(x)$ with multiple roots at the roots of unity, 
and a sequence of $q-1$ real polynomials $p_1(z), \dots, p_q(z)$ that satisfy some combinatorial constraints.
When $N > (2t+1)^2(d-1)$, we prove constructively that an uncountable number of such codes exist.

   The physical motivation of permutation-invariant codes arises from the complete immunity of quantum information stored in such codes from numerous physical processes that arise for example in many electron systems \cite{ouyang2014permutation,Rus00}. 
 Such physical processes can be described with {\em quantum permutation channels}, where every Kraus operator of such channels is a power series of a linear combination of $\varphi(P)$ for $P \in S_N$ \cite{ouyang2015permutation}.
When the quantum channel afflicting the physical system can be written as the composition of a quantum permutation channel with a quantum channel that introduces $t$ errors,
the effective quantum channel is just a quantum channel that introduces $t$ errors.  
It is in this case that permutation-invariant codes \cite{Rus00,PoR04,ouyang2014permutation,ouyang2015permutation}
 are the most natural candidates to encode quantum information.

The outline of this rest of this paper is as follows. 
In Section \ref{sec:prelims}, we introduce terminology related to ordered partitions, permutation-invariant sets, Dicke states, and quantum error correction.
In Section \ref{sec:generating-function},
we introduce a few lemmas that relate generating functions to some combinatorial identities.
In Section \ref{sec:Dicke-states-poly},
we prove that the trace of a particular rank one projector on the symmetric subspace with a linear operator can be interpreted as a polynomial.
In Section \ref{sec:poly-constructions},
we explain how the coefficients of the polynomial $f(x) = \sum_{z=0}^n f_z x^z$ and the values of the polynomials $p_1(z), \dots, p_q(z)$ for $z=0,\dots, n$ relate to our code construction in 
  Theorem \ref{thm:typeA} which constructs permutation-invariant codes encoding a single qubit and Theorem \ref{thm:typeB} which constructs permutation-invariant codes encoding a $d$-level system. Explicit examples of permutation-invariant codes that follow from these two theorems are supplied in Section \ref{sec:examples}. In Section \ref{sec:poly-constructions} we also prove that there is an uncountable number of permutation-invariant quantum codes that correct $t$ errors whenever $N > (2t+1)^2(d-1)$ in Theorem \ref{thm:uncountable}.
  In Section \ref{sec:comparisons}, we make comparisons of our code construction with prior permutation-invariant quantum codes.
  In Section \ref{sec:conclude},
  we give some concluding remarks.

\section{Preliminaries} \label{sec:prelims}
 
 \subsection{The Dirac bra-ket notation}
 Let $\mathcal K$ be a finite dimension complex Euclidean space.
 Every vector $\phi \in \mathcal K$ can be 
 written in the Dirac bra-ket notation as $|\phi\>$, which is called a {\em ket}.
 The adjoint of $\phi$ is denoted in the Dirac bra-ket notation is $\<\phi|$, and is called a {\em bra}.
 In this paper, in accordance with existing convention in quantum information theory, we use the Dirac bra-ket notation to 
 denote the basis vectors of our permutation-invariant code, 
 and also the basis vectors of the permutation-invariant space of $(\mathbb C^q) ^{\otimes N}$.

 \subsection{Ordered partitions and permutation-invariant sets}
Let $\mathbb N = \{0, 1 ,\dots, \} $ denote the set of the non-negative integers. 
We say that a vector ${\bf n} = (n_1,\dots, n_{q})$ is 
an {\em ordered partition} of a positive integer $N$ into $q$ parts 
if $n_1,\dots, n_{q} \in \mathbb N$ and $n_1 + \dots +n_{q} = N$.
Here, $n_i$ counts the number of $i$'s that appears in the integer partition of $N$.
We denote the set of ordered partitions of $N$ into $q$ parts as $T_{N,q}$.

For every $N$-tuple ${\bf c} =(c_1, \dots, c_N)  \in \mathbb \{1,\dots, q\}^N$ and $k \in \mathbb N$, 
let $\wt_k({\bf c}) = | \{ i : c_i = k \} |$ count the number of components of {\bf c} that are equal to $k$,
and let $\vec \wt_q({\bf c}) =  (\wt_1({\bf c}) , \dots, \wt_{q}({\bf c})  ) $.
We denote the multinomial coefficient that counts the number of $N$-tuples {\bf c} in $\mathbb \{1,\dots, q\}^N$ for which 
$\vec \wt_q({\bf c}) = (n_1,\dots, n_{q})$ as 
\[\binom{N}{{\bf n}} =  \frac{N!}{n_1 ! \dots n_{q}!}.\]
For every ordered partition ${\bf n}$ of $N$ into $q$ parts,
we define the permutation-invariant set of type {\bf n} as
\[C_{\bf n} = \{ {\bf c} \in \mathbb \{1,\dots, q\}^{N} : 
 \vec \wt_q({\bf c}) =  {\bf n} )  \} .\]
Given $N$-tuples ${\bf x}, {\bf y} \in \mathbb  N^N$, we denote the Hamming distance between ${\bf x}$ and ${\bf y}$ as 
\[d_H({\bf x}, {\bf y})= |\{ 1\le i \le N: x_i \neq y_i\} |.\]
Given that ${\bf n}$ and ${\bf u}$ are ordered partitions of $N$ into $q$ parts, 
we correspondingly denote the minimum Hamming distance between the elements of $C_{\bf n}$ and the elements of $C_{\bf u}$ as
\begin{align}
\Delta({\bf n},{\bf u}) = \min \{ d_H({\bf x}, {\bf y}) : {\bf x} \in  C_{\bf n}, {\bf y} \in  C_{\bf u} \}.
\end{align}
For all subsets $T$ of all ordered partitions $T \subset T_{N,q}$, we denote \begin{align}
\Delta(T) = \min \{\Delta({\bf n},{\bf u})  : {\bf n},{\bf u} \in T, {\bf n} \neq {\bf u} \}
\end{align}
as the minimum distance between the distinct permutation-invariant sets induced by $T$.
   
 \subsection{Dicke states}
The permutation-invariant codes of this paper 
are expressed as linear combinations of Dicke states, which we proceed to describe.
Given any ordered partition ${\bf n}\in T_{N,q}$, 
we use the Dirac bra-ket notation to denote the Dicke state of type {\bf n} and the unnormalized Dicke state of type {\bf n} respectively as 
\begin{align}
|D_{\bf n}\>  
&=  {\sqrt{|C_{\bf n}|^{-1} }}\sum_{{\bf c } = (c_1 ,\dots, c_N) \in C_{\bf n} } 
 {\bf e}_{c_1} \otimes \dots \otimes {\bf e}_{c_N}  
 \notag\\ 
&=  \frac{1}{\sqrt{\binom{N}{\bf n} }}\sum_{{\bf c }= (c_1 ,\dots, c_N) \in C_{\bf n} }
 {\bf e}_{c_1} \otimes \dots \otimes {\bf e}_{c_N}    ,
\end{align}
and 
\begin{align}
|H_{\bf n}\>  = \sum_{{\bf c } =(c_1 , \dots, c_N)\in C_{\bf n} } 
{\bf e}_{c_1} \otimes \dots \otimes {\bf e}_{c_N}  
 .
\end{align}
The set of Dicke states $\{ | D_{{\bf n}} \> : {\bf n} \in T_{N,q}\}$ is an orthonormal basis of the symmetric subspace of $N$ qudits.
This means that any permutation-invariant quantum state is a linear combination of these Dicke states 
$| D_{{\bf n}} \>$.
When each qudit has a dimension of $q$, the dimension of the symmetric subspace of $N$ qudits is $\binom {N+ q -1}{q-1}$, 
which is the number of ways to assign $q$ colors to $N$ unordered balls.

\section{The generating function $f(x)$} \label{sec:generating-function}
  
 We give a summation identity that is related to the coefficients of the polynomial $f(x)= \sum_{z=0}^n f_z x^z$ by interpreting $f(x)$ as an appropriate generating function in the following lemma.
\begin{lemma} \label{lem:cute-identity}
Let $\omega = e^{2 \pi i / d}$ where $d$ is a positive integer and $d \ge 2$.
Let $m$ be a positive integer, $g(x)$ be a polynomial of degree $\gamma$, and $n= \gamma+m (d-1)$.
Let $f(x) = (1+x+ \dots  + x^{d-1})^m g(x) = \sum_{z=0}^{n} f_z x^z$.
Then for every $k = 1,\dots , d-1$,
\begin{align}
\sum_{z=0}^n  f_z \omega^{k z} &=
f_0 + f_1 \omega^k + \dots + f_n \omega^{k n} =   0. \label{eq:combi-first}
\end{align}
Moreover for every $c= 1, \dots ,   m-1$
 and $k = 1,\dots , d-1$,
\begin{align}
\sum_{z=0}^n  f_z \omega^{k z} z^{c} &=  f_1 \omega^{k} 1^c+ \dots + f_n \omega^{k n } n^c =   0.\label{eq:combi-second}
\end{align}
\end{lemma}
\begin{proof} 
Since the polynomial $(1+x+\dots + x^{d-1})$ is equal to zero for every $x= \omega, \dots, \omega^{d-1}$, the polynomial $f(x)$ also evaluates to zero for every $x= \omega, \dots, \omega^{d-1}$.
 
 Let $(z)_j = \Pi_{k=1}^j (z-k+1)$ denote a falling factorial where $z$ is any real number and $j$ is any non-negative integer.
The $j$-th derivative of $f(x)$ for every $j = 1, \dots, m-1$ is
\begin{align}
\frac{ d^j}{dx^j} f(x) = 
\frac{ d^j}{dx^j}  \sum_{z=0}^n f_z x^z=
\sum_{z=0}^n f_z (z)_j x^{z-j}, \label{eq:differentiate-fx}
\end{align}
and there exists some polynomial $h(x)$ such that  
\begin{align}
\frac{ d^j }{dx^j }f(x) = (1+x + \dots + x^{d-1}) h(x). \label{eq:diff-poly}
\end{align}
Hence (\ref{eq:diff-poly}) implies that $\omega^{k j} \frac{ d^j }{dx^j } f(x)$ also evaluates to zero for $x = \omega^k$ for every $k=1, \dots, d-1$.
Hence 
\begin{align}
0 =\omega^{k j} \sum_{z=0}^n f_z (z)_j \omega^{k(z-j) }
=  \sum_{z=0}^n f_z (z)_j \omega^{kz}.
\end{align} 
For every positive integer $c$ and $j$, 
let $S(c,j)$ denote the Stirling number of the second kind \cite{van2001course} which counts the number of ways to partition the integer $c$ into exactly $j$ parts.
The monomial $z^c$ is a linear combination of the falling factorials $(z)_j$ given by $z^c= \sum_{j=0}^c S(c,j) (z)_j,$ where $(z)_0=1$.
Hence
\begin{align}
\sum_{z=0}^n f_z z^c \omega^{kz} 
&=
\sum_{z=0}^n f_z \sum_{j=0}^c S(c,j) (z)_j  \omega^{kz} \notag\\
&=
 \sum_{j=0}^c  S(c,j) \sum_{z=0}^n f_z  (z)_j  \omega^{kz}. \notag
\end{align}
Hence $0 = f(\omega^k) = \sum_{z=0}^n f_z \omega^{kz} $ for every $k = 1, \dots, d-1$,
and this proves (\ref{eq:combi-second}).
\end{proof}
Lmma \ref{lem:cute-identity} immediately implies the following identity.
\begin{lemma} \label{lem:cute-identity-0} 
Let $m$ be a positive integer, $g(x)$ be a polynomial of degree $\gamma$, and $n= \gamma+m (d-1)$.
Let $f(x) = (1-x)^m g(x) = \sum_{z=0}^{n} f_z x^z$.
Then $f_0 + \dots + f_n = 0$, and for every $c= 1, \dots , m-1$, we also have
\begin{align}
\sum_{z=0}^n  f_z  z^{c} &=  f_1  1^c+ \dots + f_n  n^c =   0.\notag
\end{align}
\end{lemma}
 \begin{proof}
Let $h(x) = f(-x)$ and $h_z = f_z (-1)^z$. Then  
 \begin{align}
f(-x) = (1+x)^m g(-x) = \sum_{z=0}^n f_z x^z (-1)^z =\sum_{z=0}^n h_z x^z = h(x) .\notag
\end{align}
Since $h(x) = (1+x)^m g(-x)$ and $g(-x)$ is just a polynomial in $x$, 
Lemma \ref{lem:cute-identity} implies that for all $c=0,\dots, m-1$,
we have $\sum_{z=0}^n h_z (-1)^z z^c = 0$. Making the substitution $h_z = f_z (-1)^z$,
we get $\sum_{z=0}^n f_z z^c = 0$, we prove the lemma.
\end{proof} 
We emphasize that in Lemma \ref{lem:cute-identity-0}, the coefficient $f_z$ multiplies with a monomial in $z$ instead of an exponential in $z$ and hence the summation is not a power series. 
We give some examples of polynomials $f(x)$ for which Lemma \ref{lem:cute-identity-0} applies.
\begin{enumerate}
\item When $f(x)=  (x-1)^5$, we have $f_5 = 1, f_4 = -5, f_3 = 10 , f_2 = -10, f_1 = 5, f_0 = -1.$ For every $j \in \{0,1,2,3,4\}$ we have
\begin{align}
(-1) 0^j + (5)1^j+(-10) 2^j + (10) 3^j  + (-5) 4^j + (1) 5^j = 0  . \notag
\end{align}
\item
When $f(x)=  (1+x)(x-1)^5$,
we have
$f_6 = 1, f_5 = -4, f_4 = 5 , f_3 = 0, f_2 = -5, f_1 = 4, f_0 = -1$.
It is easy to verify that for every $j \in \{0,1,2,3,4\}$ we have 
\begin{align}
(-1) 0^j + (4) 1^j + (-5)2^j +(5) 4^j + (-4) 5^j + (1) 6^j = 0. \notag
\end{align}
\end{enumerate} 

\section{From Dicke states to polynomials} \label{sec:Dicke-states-poly}
  Here we prove that the trace of a particular rank one projector on the symmetric subspace with a linear operator can be interpreted as a polynomial.
  In particular, given that ${\bf p}(z)$ is a tuple of polynomials, 
$\<D_{{\bf p}(z)}|P|D_{{\bf p}(z)}\>$ is always a polynomial in $z$ for every operator $P$ in 
$L((\mathbb C^q) ^{\otimes N})$. 
\begin{lemma} \label{lem:<D|P|D>=poly}
Let ${\bf p}(z) = (p_1(z) , \dots, p_q(z))$ be an ordered partition of the positive integer $N$ for all $z = 0, \dots , q-1$
and $p_1(z) , \dots, p_q(z)$ are polynomials of degree at most $\theta$.
Then for every matrix $P \in L((\mathbb C^q)^{\otimes N})$ of weight $w$,
$\<D_{{\bf p}(z)}|P| D_{{\bf p}(z)}\>$ is a polynomial of degree 
at most $w \theta$ in the variable $z$.
\end{lemma}
\begin{proof} 
 Without loss of generality, the permutation-invariance of the Dicke states $| D_{{\bf p}(z)}\>$ allows us to consider $P = E \otimes  I_{w+1,\dots,N}$ that operates non-trivially on the first $w$ qudits, where 
 $I_{w+1,\dots,N}$ is an identity matrix on the last $N-w$ qudits.
Recall that $T_{N,q}$ denotes the set of all ordered partitions of $N$ into $q$ parts.
From the Vandermonde decomposition we have
\begin{align}
| H_{ {\bf p}(z)}\> = 
\sum_{ 
	\substack {  
		{\bf a} \in T_{w,q} \\
		{\bf b} \in T_{N-w,q} \\
		{\bf a} + {\bf b} = {\bf p}(z) \\
	 }    
}  
   | H_{\bf a} \> \otimes 
 | H_{\bf b}\>,
\end{align}
where $| H_{\bf a} \>$ and $| H_{\bf b} \>$ are unnormalized Dicke states on $w$ qudits 
and $N-w$ qudits respectively. 
\begin{align}
\< H_{{\bf p}(z)} | P  | H_{{\bf p}(z)} \>
= 
\sum_{ 
	\substack {  
		{\bf a}, {\bf a}' \in T_{w,q} \\
		{\bf b}, {\bf b}' \in T_{N-w,q} \\
		{\bf a} + {\bf b} = {\bf p}(z) \\
		{\bf a}' + {\bf b}' = {\bf p}(z) \\
	 }    
}  
 \< H_{{\bf a}'} | E | H_{\bf a} \>  \< H_{{\bf b'}} | H_{\bf b} \>.
\end{align}
Since $ \< H_{{\bf b'}} | H_{\bf b} \>=  0$ whenever ${\bf b} \neq {\bf b}'$,
the above can be rewritten as
\begin{align}
\< H_{{\bf p}(z)} | P  | H_{{\bf p}(z)} \>
= 
\sum_{ 
	\substack {  
		{\bf a} \in T_{w,q} \\
		{\bf b} \in T_{N-w,q} \\
		{\bf a} + {\bf b} = {\bf p}(z) \\
	 }    
}  
 \< H_{{\bf a}} | E | H_{\bf a} \>  \< H_{{\bf b}} | H_{\bf b} \>.
\end{align}
Hence
\begin{align}
\< D_{{\bf p}(z)} | P  | D_{{\bf p}(z)} \>
&=
\< H_{{{\bf p}(z)}} | P  | H_{{{\bf p}(z)}} \> \binom {N  }  {{\bf p}(z)}^{-1} \notag\\
&= \!\!\!\! \!\!\!\!
\sum_{
	\substack {  
		{{\bf a}}  \in T_{w,q} \\
		{{\bf p}(z)} - {{\bf a}} \in T_{N-w,q} \\
	 }    
} \!\!\!\!  \!\!\!\!
 \< H_{{\bf a}} | E | H_{\bf a} \>
\binom {N-w}{ {{\bf p}(z)} - {\bf a}}  \binom {N  }{{\bf p}(z)}^{-1}. \label{eq:dicke-E-inner-product}
\end{align}
 Since $\< H_{{  {\bf a} }} | E  | H_{{ {\bf a}}} \>$ is independent of the variable $z$, 
 it suffices to show that 
every $ \binom {N-w}{ {{\bf p}(z)} - {\bf a}}  \binom {N  }{{\bf p}(z)}^{-1}$
in (\ref{eq:dicke-E-inner-product}) is a polynomial of order at most $w \theta$ in the variable $z$. 
 
For every ${\bf a} \in T_{w,q}$ and ${\bf p}(z) \in T_{N,q}$ such that
${\bf p}(z)  - {\bf a} \in T_{N-w,q}$, we have
\begin{align}
 \binom {N-w}{ {{\bf p}(z)} - {\bf a}}  \binom {N  }{{\bf p}(z)}^{-1}
&=
\frac{(N-w)!}{(p_1(z)-a_1)! \dots (p_q(z)-a_q)!}
\frac{(p_1(z))! \dots (p_q(z))!}{N!}
\notag\\ 
& = 
 \frac{ (p_1(z))_{a_1} \dots  (p_q(z))_{a_{q}}  }{ (N)_{w} }.
\end{align}
Since each $p_j(z)$ is a polynomial of order at most $\theta$, it follows that $ \binom {N-w}{ {{\bf p}(z)} - {\bf a}}  \binom {N  }{{\bf p}(z)}^{-1}$ is a polynomial of order at most $w \theta$ in the variable $z$.
\end{proof}

\section{Polynomial constructions of permutation-invariant codes} \label{sec:poly-constructions}
We provide two separate constructions of permutation-invariant codes that correct $t$ errors.
 In the first construction, we restrict our attention to permutation-invariant codes that encode a single qubit into $N$ qudits, with each qudit a $q$-level system.
 The construction relies only on the properties of the following polynomials.
 \begin{enumerate}
\item A degree $n$ polynomial $f(x)$ with real coefficients that has a root at $x=1$ with multiplicity at least $2t+1$.
\item A tuple of polynomials $(p_1(z), \dots , p_q(z))$ that is an ordered partition of $N$ into $q$ parts for every $z= 0,\dots, n$.
\end{enumerate}
As long as the permutation-invariant sets induced by the ordered partitions 
$(p_1(z), \dots , p_q(z))$ are separated by a minimum distance of at least $2t+1$, we can construct permutation-invariant codes on $N$ qudits using these polynomials as given in the following theorem.
\begin{theorem} \label{thm:typeA} 
Let $f(x) = (x-1)^m g(x) = \sum_{z=0}^n f_z x^z$ be a polynomial of degree $n$,
where $g(x)$ is a polynomial with real coefficients.
Let $p_1(z) , \dots, p_q(z)$ be polynomials  
such that for all $z = 0, \dots , n$, ${\bf p}(z) = (p_1(z),\dots,p_q(z))$ is an ordered partition of the positive integer $N$.
Let the logical zero and the logical one be given respectively by the states 
\begin{align}
|0_L\> &= 
\sqrt 2 (|f_0| + \dots + |f_n|)^{-1/2}
\sum_{\substack{z=0,\dots , n\\ f_z > 0}}  \sqrt{  f_z  }  
| D_{{\bf p}(z)}  \>  ,
\notag \\
|1_L\> &= 
\sqrt 2  (|f_0| + \dots + |f_n|)^{-1/2}
\sum_{\substack{z=0,\dots , n\\ f_z < 0}}  \sqrt{ -f_z }
| D_{{\bf p}(z)}  \>  \label{eq:picodes-typeA}.
\end{align} 
Suppose that $\Delta(\{ {\bf p}(z): z = 0, \dots, n \}) \ge 2t+1$ 
and that the degree of the polynomials $p_1(z) , \dots, p_q(z)$ is at most $\frac{m-1}{2t}$,
where $t$ is a positive integer.
Then $\{|0_L\> , |1_L\> \}$ is an orthonormal basis and spans a code that corrects $t$ errrors.
\end{theorem}
Theorem \ref{thm:typeA} can be seen as a consequence of the Knill-Laflamme error correction criterion in Theorem \ref{thm:KL} and the Lemmas \ref{lem:cute-identity} and Lemma \ref{lem:<D|P|D>=poly}. 
In its proof, the orthogonality condition of Theorem \ref{thm:KL} is trivially satisfied, 
and the non-deformation condition of Theorem \ref{thm:KL} holds because of the aforementioned lemmas.
\begin{proof}[Proof of Theorem \ref{thm:typeA}]
Since $P \in L((\mathbb C^q) ^{\otimes N})$ operates non-trivially on at most $2t$ qudits and the codes $C_{{\bf p}(z)}$ have a mutual minimum distance of at least $2t+1$, the orthogonality condition of Theorem \ref{thm:KL} holds.

To show that the non-deformation conditions of Theorem \ref{thm:KL} hold, 
it suffices to show that 
\begin{align}
\<0_L| P |0_L\> =  \sqrt 2 (|f_0| + \dots + |f_n|)^{-1}
\sum_{\substack{z=0,\dots , n\\ f_z > 0}}  f_z 
\<| D_{{\bf p}(z)}|P | D_{{\bf p}(z)}  \>  ,
\end{align}
and
\begin{align}
\<1_L| P |1_L\> = \sqrt 2 (|f_0| + \dots + |f_n|)^{-1}
\sum_{\substack{z=0,\dots , n\\ f_z < 0}}  -f_z  
\<| D_{{\bf p}(z)}|P | D_{{\bf p}(z)}  \>  .
\end{align} 
Now
\begin{align}
\<0_L|P|0_L\>-\<1_L|P|1_L\> &=
\sqrt 2 \sum_{ z=0,\dots,n }  f_z   \<D_{{\bf p}(z)} |P| D_{{\bf p}(z)}\>.
\end{align} 
Lemma~\ref{lem:<D|P|D>=poly} implies that the polynomials $\< D_{{\bf p}(z)} | P | D_{{\bf p}(z)} \> $ have degree no more than $2t \theta$ in the variable $z$.
  Hence for some constants $\alpha_c \in \mathbb C$
\begin{align}
\sum_{z =0}^n  f_z   	\< D_{{\bf p}(z)} | P | D_{{\bf p}(z)} \>   
=&
\sum_{z =0}^n  f_z   	 \sum_{c=0}^{2t \theta } \alpha_c z^c \notag\\
=&
	 \sum_{c=0}^{2t \theta } \alpha_c \left( \sum_{z =0}^n  f_z  z^c \right).
	 \label{eq:non-deformation-proof}
\end{align}
The bracketed term above is zero because of Lemma \ref{lem:<D|P|D>=poly}.
Hence $\<0_L| P |0_L\> = \<1_L| P |1_L\> $ for every $P \in L((\mathbb C^q) ^{\otimes N})$ that operates non-trivially on at most $2t$ qudits.
This proves that the non-deformation condition of Theorem \ref{thm:KL} holds.

Taking $P$ to be the identity operator then implies that $ \<0_L| 0_L\>  = \<1_L|  1_L\> $ which proves that $ |0_L\>$ and $ |1_L\> $ both have unit norm. Since $ |0_L\>$ and $ |1_L\> $ are trivially orthogonal, $\{|0_L\>,|1_L\>\}$ is an orthonormal basis.
\end{proof}

 In the second construction, we construct permutation-invariant codes that encode a $d$-level system into $N$ qudits. We rely on the properties of the following polynomials.
 \begin{enumerate}
\item A degree $n$ polynomial $f(x)$ with non-negative coefficients that divides $(1+x+\dots +x^{d-1})^{2t+1}$ to yield a polynomial.
\item A tuple of polynomials $(p_1(z), \dots , p_q(z))$ that is an ordered partition of $N$ into $q$ parts for every $z= 0,\dots, n$.
\end{enumerate}

\begin{theorem} \label{thm:typeB} 
Let $f(x) = (1+x+\dots +x^{d-1})^m g(x)$ be a polynomial of degree $n$ such that $f(x)$ has non-negative coefficients.
Let $p_1(z) , \dots, p_q(z)$ be polynomials  
such that for all $z = 0, \dots , n$, ${\bf p}(z) = (p_1(z),\dots,p_q(z))$ is an ordered partition of the positive integer $N$.
For $k=0,\dots, d-1$, let
\begin{align} 
|k_L\> &= 
\sqrt{d} (f_0 + \dots + f_n)^{-1/2}
\sum_{\substack{z=0,\dots , n\\ (z-k)/d \in \mathbb N}}   \sqrt{ f_z }
| D_{{\bf p}(z)}  \> . \label{eq:picodes-typeB}
\end{align} 
Suppose that $\Delta(\{ {\bf p}(z): z = 0, \dots, n \}) \ge 2t+1$ 
and that the degree of the polynomials $p_1(z) , \dots, p_q(z)$ is at most $\frac{m-1}{2t}$,
where $t$ is a positive integer.
Then $\{|k_L\> : k =0,\dots, d-1 \}$ is an orthonormal basis and spans a code that corrects $t$ errrors.
\end{theorem}

The proof of Theorem \ref{thm:typeB} is similar to the proof of Theorem \ref{thm:typeA}, which uses Theorem \ref{thm:KL} with the Lemmas \ref{lem:cute-identity-0} and Lemma \ref{lem:<D|P|D>=poly}.
 In this proof, the orthogonality condition of Theorem \ref{thm:KL} is no longer trivially satisfied and relies on the aforementioned lemmas.
 The non-deformation condition of Theorem \ref{thm:KL} on the other hand is satisfied trivially.
\begin{proof}[Proof of Theorem \ref{thm:typeB}]
Using the Dirac bra-ket notation, we denote the discrete Fourier transform of the vectors $|k_L\>$ as 
\begin{align} 
|\tilde k\> &=
\frac{1}{\sqrt{d}} \sum_{j=0}^{d-1} \omega^{j} |j_L\>
  =  (f_0 + \dots + f_n)^{-1/2}
\sum_{z=0}^n  \omega^{k z}  \sqrt{ f_z }
| D_{{\bf p}(z)}  \>  \label{eq:picodes-typeB}.
\end{align} 
To show that 
$\{|k_L\> : k = 0, \dots, d-1\}$ is an orthonormal basis,
it suffices to prove that $\{|\tilde k\> : k = 0, \dots, d-1\}$ is an orthonormal basis since the Fourier transform is a unitary transformation.
Clearly $|\tilde k\>$ has unit norm for all $k = 0,\dots, d-1$. Hence it remains to demonstrate that $|\tilde j\>$ and $|\tilde k\>$ are orthogonal for distinct $j$ and $k$.
It suffices then to show that the more general orthogonality condition $ \<\tilde k |P|\tilde j\> =0 $ of Theorem \ref{thm:KL} holds whenever $P$ is an $N$ qudit operator that acts non-trivially on at most $2t$ qudits. 
Note that 
\begin{align}
\<\tilde k |P|\tilde j\> 
  =  (f_0 + \dots + f_n)^{-1}
\sum_{z=0}^n \sum_{y=0}^n  \omega^{-k z+j y} 
\sqrt{f_z f_y}\< D_{{\bf p}(z)} |P | D_{{\bf p}(y)}  \> . \label{eq:above}
\end{align}
Since $P$ has a weight of at most $2t$, Dicke states of distinct types do not overlap under $P$, and 
\begin{align}
\<\tilde k |P|\tilde j\> 
  =  (f_0 + \dots + f_n)^{-1}
\sum_{z=0}^n  \omega^{(j-k) z} 
f_z \< D_{{\bf p}(z)} |P | D_{{\bf p}(z)}  \> .
\end{align}
Since $\< D_{{\bf p}(z)} |P | D_{{\bf p}(z)}  \>$ is a low order polynomial,
the combinatorial identity in Lemma \ref{lem:cute-identity} implies that
$\<\tilde k |P|\tilde j\>  = 0$ whenever $j \neq k$, which is the orthogonality condition of Theorem \ref{thm:KL}. 
The non-deformation condition of Theorem \ref{thm:KL} holds trivially because 
for every $k= 0,\dots, d-1$,
\begin{align}
\<\tilde k |P|\tilde k\> 
	  =  (f_0 + \dots + f_n)^{-1}
\sum_{z=0}^n  f_z \< D_{{\bf p}(z)} |P | D_{{\bf p}(z)}  \> .
\end{align}
\end{proof}

\section{Examples of permutation-invariant codes} \label{sec:examples}
 
   We now supply a few examples of permutation-invariant codes where the weight distribution for the Dicke states for the permutation-invariant code is linearly shifted, and the square of the amplitudes do not follow the binomial distribution.
 \begin{example} 
 With the construction of Theorem \ref{thm:typeA}, $N=19$ qubits are used with the polynomial $f(x) =(1+x)(x-1)^5$ and the tuple of polynomials ${\bf p}(z)=( N - 1-3z, 1 + 3z)$ to obtain a permutation-invariant code encoding one qubit that corrects one arbitrary error. The orthonormal basis vectors are
 \begin{align}
  |0_L\> = \frac{\sqrt 4 |D_{(15,4)}\> + \sqrt 5 |D_{(6,13)}\> + |D _{(0,19)}\>}{\sqrt{10}} 
   ,\quad  
   |1_L\> = \frac{ |D_{(18,1)}\> +  \sqrt {5} |D_{(12,7)}\> + \sqrt 4 |D_{(3,16)}\>}{\sqrt{10}}. \notag
 \end{align}
 \end{example}
 
  \begin{example} 
   With the construction of Theorem \ref{thm:typeA}, $N= 108$ qutrits are used with the polynomial $f(x) = (1+x)(x-1)^5$ and the tuple of polynomials ${\bf p}(z) = (N-3z^2, 3z^2, 0)$
   to obtain a permutation-invariant code encoding one qubit that corrects one arbitrary error.
The orthonormal basis vectors are
\begin{align}
  |0_L\> &= \frac{
  \sqrt 4 |D_{(105,3,0)}\> + \sqrt 5 |D_{(60,48,0)}\> + |D_{(0,108,0)}\>
  }{\sqrt{10}} \notag\\
 |1_L\> &= \frac{
 |D_{(108, 0,0)} \> +  \sqrt {5} |D_{(96,12,0)}\> + \sqrt 4 |D_{(33,75,0)}]\>
 }{\sqrt{10}} .\notag
\end{align}
  \end{example}
 
     \begin{example} \label{example:d-level}
  Let $N=(2t+1)^2(d-1)$, $f(x) = (1+x+\dots +x^{d-1})^{2t+1} = \sum_{z=0}^{(2t+1)(d-1)} f_z x^z$, 
  $p_1(z) = (2t+1)z, p_2(z) = N- p_1(z),$ and $p_3(z) , \dots, p_q(z) = 0$.
The construction of Theorem \ref{thm:typeB} yields an $N$-qudit permutation invariant code that encodes a $d$-level system into $(2t+1)^2(d-1)$ qudits and can correct $t$ errors.
The orthonormal basis vectors are $\{|k_L\> , k = 0,\dots, d-1\}$ where
\begin{align}
|k_L\> &= 
\sqrt{d^{-m+1}} \sum_{\substack{z=0,\dots , (2t+1)(d-1)\\ (z-k)/d \in \mathbb N}}   \sqrt{ f_z }
| D_{{\bf p}(z)}  \> .\notag
\end{align} 
 \end{example}
We now give examples of permutation-invariant codes on qubits that correct a single error while encoding a 3-level system, a 4-level system and a 5-level system that are all based on Example \ref{example:d-level}.
  \begin{example} \label{eg:3level}
  Let $N=18$, $f(x) = (1+x+x^2)^{3} = \sum_{z=0}^{6} f_z x^z$, 
  $p_1(z) =3z, p_2(z) = N- 3z$.
The construction of Theorem \ref{thm:typeB} yields an $18$-qubit permutation invariant code that encodes a $3$-level system and corrects 1 error.
The orthonormal basis vectors are  
\begin{align}
|0_L\> &= \frac{| D_{(18,0)}\> + \sqrt{7}| D_{(9,9)}  \> + | D_{(0,18)}  \>}{3} ,\notag\\
|1_L\> &= \frac{  \sqrt{3}| D_{(15,3)}  \> + \sqrt{6} | D_{(6,12)} \>}{3},  \notag\\
|2_L\> &= \frac{  \sqrt{6}| D_{(12,6)}  \> + \sqrt{3} | D_{(3,15)}  \>}{3} 
 .\notag
\end{align} 
 \end{example}
  \begin{example} \label{eg:4level}
  Let $N=27$, $f(x) = (1+x+x^2+x^3)^{3} = \sum_{z=0}^{9} f_z x^z$, 
  $p_1(z) =3z, p_2(z) = N- 3z$.
The construction of Theorem \ref{thm:typeB} yields a $27$-qubit permutation invariant code that encodes a $4$-level system and corrects 1 error.
The orthonormal basis vectors are  
\begin{align}
|0_L\> &= \frac{| D_{(27,0)}\> + \sqrt{12}| D_{(15,12)}  \> + \sqrt{3} | D_{(3,24)}  \>}{4} ,\notag\\
|1_L\> &= \frac{  \sqrt{3}| D_{(24,3)}  \> + \sqrt{12} | D_{(12,15)} \>+ | D_{(0,27)}  \>}{4},  \notag\\
|2_L\> &= \frac{  \sqrt{6}| D_{(21,6)}  \> + \sqrt{10} | D_{(9,18)} \>}{4},  \notag\\
|3_L\> &= \frac{  \sqrt{10}| D_{(18,9)}  \> + \sqrt{6} | D_{(6,21)} \>}{4}   .\notag
\end{align} 
 \end{example}

  \begin{example} \label{eg:5level}
  Let $N=36$, $f(x) = (1+x+x^2+x^3+x^4)^{3} = \sum_{z=0}^{12} f_z x^z$, 
  $p_1(z) =3z, p_2(z) = N- 3z$.
The construction of Theorem \ref{thm:typeB} yields a $36$-qubit permutation invariant code that encodes a $5$-level system and corrects 1 error.
The orthonormal basis vectors are  
\begin{align}
|0_L\> &= \frac{| D_{(36,0)}\> + \sqrt{18}| D_{(21,15)}  \> + \sqrt{6} | D_{(6,30)}  \>}{5} ,\notag\\
|1_L\> &= \frac{  \sqrt{3}| D_{(33,3)}  \> + \sqrt{19} | D_{(18,18)} \>+\sqrt 3 | D_{(3,33)}  \>}{5},  \notag\\
|2_L\> &= \frac{  \sqrt{6}| D_{(30,6)}  \> + \sqrt{18} | D_{(15,21)} \>+ | D_{(0,36)} \>}{5},  \notag\\
|3_L\> &= \frac{  \sqrt{10}| D_{(27,9)}  \> + \sqrt{15} | D_{(12,24)} \>}{5}   ,\notag\\
|4_L\> &= \frac{  \sqrt{15}| D_{(24,12)}  \> + \sqrt{10} | D_{(9,27)} \>}{5}   .\notag
\end{align} 
 \end{example}
 At this point, we remark that the coefficients in the orthonormal basis vectors of the permutation-invariant codes supplied in Examples \ref{eg:3level}, \ref{eg:4level} and \ref{eg:5level}
 are identical to the coefficients of the orthonormal basis vectors of the binomial bosonic codes in \cite{BinomialCodes2016}.
Since the error model considered for the binomial bosonic codes is more restricted than the error model we consider, to prove that the binomial bosonic codes work, one only needs to demonstrate the orthogonality property of the Knill-Laflamme error correction criterion in Theorem \ref{thm:KL}.
In our situation, we also need to prove that the non-deformation conditions in Theorem \ref{thm:KL} hold.

We now present a result that is in the same spirit as the result on permutation-invariant quantum codes in Ref.~\cite{PoR04}, 
where an uncountable number of codes correcting a single error on 9 qubits was demonstrated, 
and we essentially rely on the construction of permutation-invariant quantum codes from Theorem \ref{thm:typeB}.
\begin{theorem}  \label{thm:uncountable}  
Let $t$ and $d$ be positive integers with $d \ge 2$, and let $N$ be an integer such that $N \ge (2t+1)^2(d-1)$. 
Then there is a permutation-invariant code on $N=(2t+1)^2(d-1)$ qudits with dimension $d$ that correct $t$ errors.
Furthermore if $N > (2t+1)^2(d-1)$,
there is an uncountable number of permutation-invariant codes on $N$ qudits with dimension $d$ that correct $t$ errors.
\end{theorem}
 \begin{proof}
 The first result 
 follows directly from Example \ref{example:d-level}. It remains to prove the second result.
We consider the codes of given by Theorem \ref{thm:typeB},
 with $f(x) = f_\theta(x)$ where
 \[ f_\theta(x) = (1+x+\dots + x^{d-1})^m (\cos^2 \theta + x \sin^2 \theta),\] 
 for $0\le \theta  \le \pi /2$,
 and any choice of the $q$-tuple of polynomials ${\bf p}(z)$ such that $\Delta(\{ {\bf p}(z) : z= 0,\dots, n\}) \ge 2t+1$.
 
Now $\sum_{z=0}^n f_z= f(1) = d^m ( \cos^2 \theta +  \sin^2 \theta) = d^m.$
 Since the orthonormal basis vectors $|k_L\>$ of Theorem \ref{thm:typeB} have a unit norm,
 this implies that 
$ \sum_{\substack{z=0,\dots , n\\ (z-k)/d \in \mathbb N}}   \sqrt{ f_z } = 
\frac{1}{d}\sum_{z=0}^n f_z$ for every $k=0,\dots d-1$, and hence $f_0 + f_d + f_{2d} + \dots  = d^{m-1}$. 
Hence the logical zero of our code can be written as 
\begin{align}
|0_{\theta}\>  = \sqrt{d^{-m+1}}
\sum_{\substack{z=0,\dots , m(d-1)\\  z /d \in \mathbb N}} 
\sqrt{f_{\theta,z}} |D_{{\bf p}(z)} \>\notag,
\end{align}
and the subscript in $|0_{\theta}\>$ makes explicit the dependence of the logical zero with the parameter $\theta$, and $f_{\theta}(x) = \sum_{z=0}^{m(d-1)} f_{\theta,z} x^z$.
For all values of $\theta$ and $\phi$ in $[0,\pi/2]$, $|0_{\theta}\>$ is orthogonal to 
$|k_{L}\>$ for all $k=1,\dots, d-1$. Hence it suffices to show that 
$0\le \<0_{\theta} |0_{\phi}\> < 1$ for all distinct values of $\theta$ and $\phi$ in $[0,\pi/2]$.

For distinct values of $\theta$ and $\phi$ in $[0,\pi/2]$, the values $x = \frac{ \cos^2 \theta}{d^{m-1}}$ and $y = \frac{\cos^2 \phi}{d^{m-1}}$ are distinct.
Note that
\begin{align}
\<0_{\theta} |0_{\phi}\>
&=
 d^{-m+1}  \sum_{z | d} \sqrt{f_{\theta,z} f_{\phi,z} } \notag\\
&=
  \sqrt {x y} +  d^{-m+1} \sum_{z | d, z \ge d} 
 \sqrt{f_{\theta,z} f_{\phi,z} }   \notag\\
 &\le 
  \sqrt {x y} +  \sqrt{(1-x)(1-y)} = 
  (\sqrt{x},\sqrt{1-x} ) \cdot (\sqrt y , \sqrt{1-y} )   \notag
\end{align} 
Since $0 \le x,y \le  1$, the above dot product is non-negative.
The vectors $ (\sqrt{x},\sqrt{1-x} ) $ and $ (\sqrt{y},\sqrt{1-y} ) $ both have unit norm, and hence the 
Cauchy-Schwarz inequality implies that their dot product is at most one.
Furthermore, the Cauchy-Schwarz inequality for the dot product between
$ (\sqrt{x},\sqrt{1-x} ) $ and $ (\sqrt{y},\sqrt{1-y} ) $ is a strict inequality since $x$ and $y$ are distinct.
Hence $0\le \<0_{\theta} |0_{\phi}\> < 1$, and this completes the proof. 
 \end{proof}

\section{Comparison with previous permutation-invariant codes} \label{sec:comparisons}
Previously constructed permutation-invariant codes have
been restricted to systems comprised of solely qubits, with Hilbert space $(\mathbb C^2)^{\otimes N}$.
Let $g, n, N$ be integers such that $g,n \ge 2t+1$ and $N \ge gn$.
Then the orthonormal basis vectors of the permutation-invariant codes that correct $t$ errors given by Ref.~\cite{ouyang2014permutation} generalizing the 9-qubit Ruskai code \cite{Rus00} have the orthonormal basis vectors
\begin{align}
|0_L\> &= \sum_{z =0,\dots,\floor{n/2} } \sqrt{\binom n {2z} 2^{-n+1}} 
|D_{(2gz,N-2gz)}\> \notag\\
|1_L\> &= \sum_{z =0,\dots,\floor{n-1/2} } \sqrt{\binom n {2z+1} 2^{-n+1}} 
|D_{(g(2z+1),N-g(2z+1))}\>. 
\label{eq:PI-codes}
\end{align}
In \cite{ouyang2014permutation}, the orthonormal basis vectors in Eq.~(\ref{eq:PI-codes}) are superpositions over Dicke states with amplitudes proportional to the square root of a binomial coefficient, where these Dicke states have weights spaced a constant number apart.
In \cite{ouyang2015permutation}, the authors proved the possibility of encoding more than a single qubit into a permutation invariant code with orthonormal basis vectors all of the form given by Eq.~(\ref{eq:PI-codes}).
However in this construction, the correction of only a single amplitude damping error is possible.

\section{Concluding remarks} \label{sec:conclude}
 In this paper, we construct permutation-invariant codes from certain polynomials, and thereby generalize the construction of the permutation-invariant codes that rely on a binomial distribution  \cite{ouyang2014permutation,Rus00} to those that rely on more general distributions.
From previous constructions of permutation-invariant codes, 
there is only a finite number of permutation-invariant-quantum codes of a fixed length;
here we show an uncountable number of permutation-invariant codes on $N$ qudits that correct $t$ errors and encode a $d$-level system exist, given that $N$ is sufficiently large.
It seems likely that the results in this paper can be combined with the technique of pasting permutation-invariant codes \cite{ouyang2015permutation} to construct other permutation-invariant codes with modest error correction capabilities.
However it remains a open problem to generalize the seven qubit permutation-invariant codes of Pollatsek and Ruskai \cite{PoR04}.

 Note that for the permutation-invariant codes in \cite{Rus00,ouyang2014permutation}, 
 the orthonormal basis vectors necessarily have amplitudes that are proportional to the square root of the binomial distribution, and the weights of the Dicke states are spaced an equal distance apart. These two properties need not hold in the permutation-invariant codes of Theorem \ref{thm:typeA}.

 \section{Acknowledgments}
 
Y.O acknowledges Tommaso Demarie for his comments on an earlier version of the manuscript.
This research was supported by the Singapore National Research Foundation under NRF Award No. NRF-NRFF2013-01.

\bibliography{../../../mybib}{}

\bibliographystyle{alpha}
\end{document}